\documentclass[11pt,a4paper,twocolumn]{article} 

\usepackage{microtype}
\usepackage{graphicx}
\usepackage{booktabs} % for professional tables
\usepackage{bm}
\usepackage{multirow}
\usepackage{algorithm}
\usepackage{algpseudocode}
\usepackage{subcaption}

\usepackage{amsmath}
\usepackage{amssymb}
\usepackage{mathtools}
\usepackage{amsthm}

\usepackage{dsfont}

\usepackage[hidelinks, colorlinks=true, allcolors=blue]{hyperref}

\usepackage[capitalize,noabbrev]{cleveref}
\usepackage[square,numbers]{natbib}

\theoremstyle{plain}
\newtheorem{theorem}{Theorem}

\newtheorem{lemma}[theorem]{Lemma}

\theoremstyle{definition}

\theoremstyle{remark}

\newcommand{\yrcite}[1]{\cite{#1}}

\usepackage{fancyhdr}

\title{\vspace{-2cm}\rule{\textwidth}{2pt} \\[2ex]
\textbf{Non-Hemolytic Peptide Classification Using \\
A Quantum Support Vector Machine} \\[2ex]
\rule{\textwidth}{2pt} }

\author{Shengxin Zhuang$^{1,5,\dagger,}$\thanks{
  Corresponding authors: \\
  shengxin.zhuang@research.uwa.edu.au, \\ \{john.tanner, du.huynh\}@uwa.edu.au, \\ frederic.cadet.run@gmail.com \\[0.5ex]
  \indent$^\dagger$Zhuang and Tanner contributed equally to this work. \\[0.5ex]
  \indent Zhuang and Tanner are supported by the Australian RTP scholarships at UWA; Wu is supported by the China Scholarship Council (Grant No.~202006470011). In addition, Zhuang is also supported by a PhD grant from the Region Reunion and European Union (FEDER) under the European Operational Program FEDER-Reunion 2021-2027. X.~F.~Cadet is supported by the UKRI CDT in AI for Healthcare \url{http://ai4health.io} (Grant No. P/S023283/1), UK.},
John Tanner$^{1,\dagger,*}$,
Yusen Wu$^1$,
Du Q.~Huynh$^{2,*}$,
Wei Liu$^2$, 
Xavier F.~Cadet$^3$,\\[0.5ex]
Nicolas Fontaine$^4$, 
Philippe Charton$^{5,6}$, 
Cedric Damour$^7$, 
Frederic Cadet$^{4,5,6,*}$, 
Jingbo Wang$^1$ 
\\[1ex]
\small\textit{$^1$Department of Physics, The University of Western Australia, Australia} \\
\small\textit{$^2$Department of Computer Science and Software Engineering, The University of Western Australia, Australia} \\
\small\textit{$^3$Department of Computing, Imperial College London, London, United Kingdom} \\
\small\textit{$^4$PEACCEL, Artificial Intelligence Department, AI for Biologics, Paris, France} \\
\small\textit{$^5$University of Paris City \& University of Reunion, BIGR, Inserm, UMR\_S1134, Paris, France} \\
\small\textit{$^6$Laboratory of Excellence GR-Ex, Paris, France} \\
\small\textit{$^7$EnergyLab, EA 4079, Faculty of Sciences and Technology, University of Reunion, Saint-Denis, France}
}

\date{February 2024}  % use \date{} if we don't want the date

\begin{document}

\maketitle

\begin{abstract}
Quantum machine learning (QML) is one of the most promising applications of quantum computation. However, it is still unclear whether quantum advantages exist when the data is of a classical nature and the search for practical, real-world applications of QML remains active. In this work, we apply the well-studied quantum support vector machine (QSVM), a powerful QML model, to a binary classification task which classifies peptides as either \emph{hemolytic} or \emph{non-hemolytic}. Using three peptide datasets, %with varying degrees of linear separability, 
we apply and contrast the performance of the QSVM, numerous classical SVMs, and the best published results on the same peptide classification task, out of which the QSVM performs best. The contributions of this work include (i) the first application of the QSVM to this specific peptide classification task, (ii) an explicit demonstration of QSVMs outperforming the best published results attained with classical machine learning models on this classification task and (iii) empirical results showing that the QSVM is capable of outperforming many (and possibly all) classical SVMs on this classification task. This foundational work paves the way to verifiable quantum advantages in the field of computational biology and facilitates safer therapeutic development. 
\end{abstract}

\section{Introduction}
Peptides, which are short chains of amino acids, contribute to a wide range of biological functions. From regulating metabolism and immune responses to playing key roles in neurological processes and tissue repair, their diverse functionalities make them attractive candidates for therapeutic development. Drug-based peptides offer several advantages over conventional small-molecule drugs, including higher specificity, lower toxicity, and enhanced bioavailability~\cite{wang2022therapeutic,chen2023accelerating,lau2018therapeutic}. Additionally, their relatively short sequences allow for easier synthesis and modification, facilitating the production of tailored structures with targeted effects. 

Despite their advantages, some peptides present biological dangers. In particular, hemolytic peptides are a class of peptides that have the ability to disrupt the cell membranes of erythrocytes, causing hemolysis, or the breakdown of red blood cells. These peptides are typically cationic and amphiphilic, meaning they have both hydrophobic and hydrophilic regions. This allows them to interact with the phospholipid bilayer of the cell membrane, causing it to destabilise and rupture. For this reason, determining whether a peptide is hemolytic or non-hemolytic is crucial, especially in therapeutic contexts~\cite{yaseen2021hemonet}. In particular, other than hemolysis, hemolytic peptides can also cause severe side effects such as anemia and kidney failure, so identifying them early in the development process can help to prevent adverse events in clinical trials and patient use. In contrast, non-hemolytic peptides demonstrate superior efficacy, making them better candidates for therapeutic purposes as they are less likely to cause unintended effects. Additionally, knowledge of hemolytic activity can guide the design of therapeutic peptides, allowing researchers to optimise their properties for safety.
Consequently, the development of accurate and efficient methods for classifying peptides as hemolytic or non-hemolytic is of great importance.

Classical machine learning (ML) approaches have been instrumental in peptide classification and screening, enabling the identification of peptides with specific properties~\cite{barman2023strategic,shiammala2023exploring}. However, these methods often face limitations in handling complex peptide data and capturing intricate relationships between peptide sequences and their activities~\cite{wan2022deep,fernandez2023autopeptideml,lv2023tppred,basith2020machine,zhang2023deep}, which creates a need for alternative classification methods. 
Recent pioneering experiments on quantum computer processors have demonstrated significant quantum computational advantages in problems such as random state sampling~\cite{boixo2018characterizing,arute2019quantum,zhong2020quantum} and density matrix learning~\cite{huang2022quantum}. Given the significance of these outcomes, quantum machine learning (QML) algorithms are expected to be capable of both recognising more complex patterns than classical ML algorithms~\cite{liu2021rigorous, wu2023quantum}, and training ML models more efficiently~\cite{rebentrost2014qsvm, li2019sublinear}. As a result, QML algorithms may provide the required alternative to classical ML by exploiting the principles and unique aspects of quantum mechanics to facilitate more accurate and efficient peptide classification and screening~\cite{ayuba_kelvin_tera_jessica_2023,avramouli2022quantum,fedorov2022quantum}.

In this work we demonstrate, via numerical simulations, the possibility of QML algorithms outperforming classical ML models on certain peptide classification tasks. In particular, we apply quantum kernel methods (QKMs)~\cite{schuld2021supervised}, specifically the quantum support vector machine (QSVM)~\cite{havlicek2019supervised}, to classify peptides as either \emph{hemolytic} or \emph{non-hemolytic}. Using three different peptide datasets, each containing labelled data which has been verified in a wet lab, we contrast the performance of the QSVMs, classical SVMs and the best results available in relevant literature for the same peptide classification task. We observe that for two of the more challenging datasets, the QSVMs were able to achieve greater accuracies.  The contributions of this work include:
\begin{enumerate}
    \item The first application of QML models to a binary classification task involving classifying peptides based on their hemolytic activity, 
    \item The first explicit demonstration of a QSVM outperforming the best published results attained with classical machine learning models on a peptide classification task, and
    \item The first instance of a peptide classification task for which QML approaches outperform many, if not all, classical SVMs.
\end{enumerate}  
We anticipate that extensions of this research using similar QML methods for other problem instances are possible, and that this work could serve as a significant turning point, facilitating potential near-future breakthroughs in the field of computational biology.

\section{Background and Related Work}

This section provides a comprehensive overview of the current state of research in the classification and screening of non-hemolytic peptides using ML and deep learning.

ML and deep learning techniques have emerged as promising tools for classifying peptides into hemolytic and non-hemolytic categories. These approaches involve training algorithms on datasets of known peptide sequences, which are labelled according to their hemolytic properties, in order to predict the hemolytic nature of novel peptides. Not only can the algorithms analyse large quantities of peptide data, they also help to identify patterns and features that are associated with hemolytic activity.

Numerous studies have demonstrated that ML and deep learning approaches can achieve high levels of accuracy in classifying peptides as hemolytic or non-hemolytic. For instance: 
Timmons et al.~\yrcite{timmons2020happenn} proposed  an artificial neural network classifier for the prediction of hemolytic activity for therapeutic peptides, which achieved an accuracy of 84.06\% on the test set.
Plisson et al.~\yrcite{plisson2020machine} tried 14 binary classifiers for predicting hemolytic activity using 3 datasets, HemoPI-1, HemoPI-2, and HemoPI-3, based on 56 sequence-based physicochemical descriptors, achieving 92.4\%, 72.3\%, and 73.2\% accuracy on the test sets respectively.
Using a recurrent neural network classifier Capecchi et al.~\yrcite{capecchi2021machine} had the an overall accuracy performance for hemolysis prediction of 76\% on the test set.
Salem et al.~\yrcite{salem2022ampdeep} leveraged transfer learning to overcome the challenge of small data and used a deep learning based model, AMPDeep, to achieve an accuracy of 86\% for hemolysis activity classification of antimicrobial peptides on the test set.
Ansari and White \yrcite{ansari2023serverless} used three bidirectional recurrent neural networks sequence-based classifiers to predict hemolysis and achieved an 84\% accuracy.
Perveen et al.~\yrcite{perveen2023hemolytic} proposed a ML-based predictor for hemolytic proteins using position and composition-based features, achieving a 91.51\% accuracy on the test set. To the best of our knowledge, the only other work which has applied QML to a computational biology task involving peptides is available on arXiv~\cite{london2023peptide}. In the said work, the authors were not able to outperform the classical models with their QML models, nor did they compare with the best published results on the task.

The current level of accuracy in classifying peptides into hemolytic and non-hemolytic categories varies from 72.3\% to 92.4\% \cite{plisson2020machine} on the test sets, depending on the algorithm and dataset used. Common performance metrics employed in these studies include accuracy, precision, recall, and the area under the receiver operating characteristic curve (AUC-ROC). These metrics provide a comprehensive evaluation of model performance, considering both false positives and false negatives. In this work though, we will use just the metric of accuracy on the test set to assess the performance of different models.

\section{Methodology}

In this section we introduce QKMs and QSVMs as they appear in the literature and describe how they are applied to the peptide classification task, including details relating to the specific models we use. Note that this section includes basic ideas and concepts from quantum computing. If the reader should require further details or explanations on the topic we recommend referring to Nielsen \& Chuang \yrcite{nielsen2000quantum}, in particular chapters 1, 2, and 4.

\subsection{Quantum Kernel Methods}

\begin{figure}[t!]
    \centering
    \includegraphics[width=\linewidth]{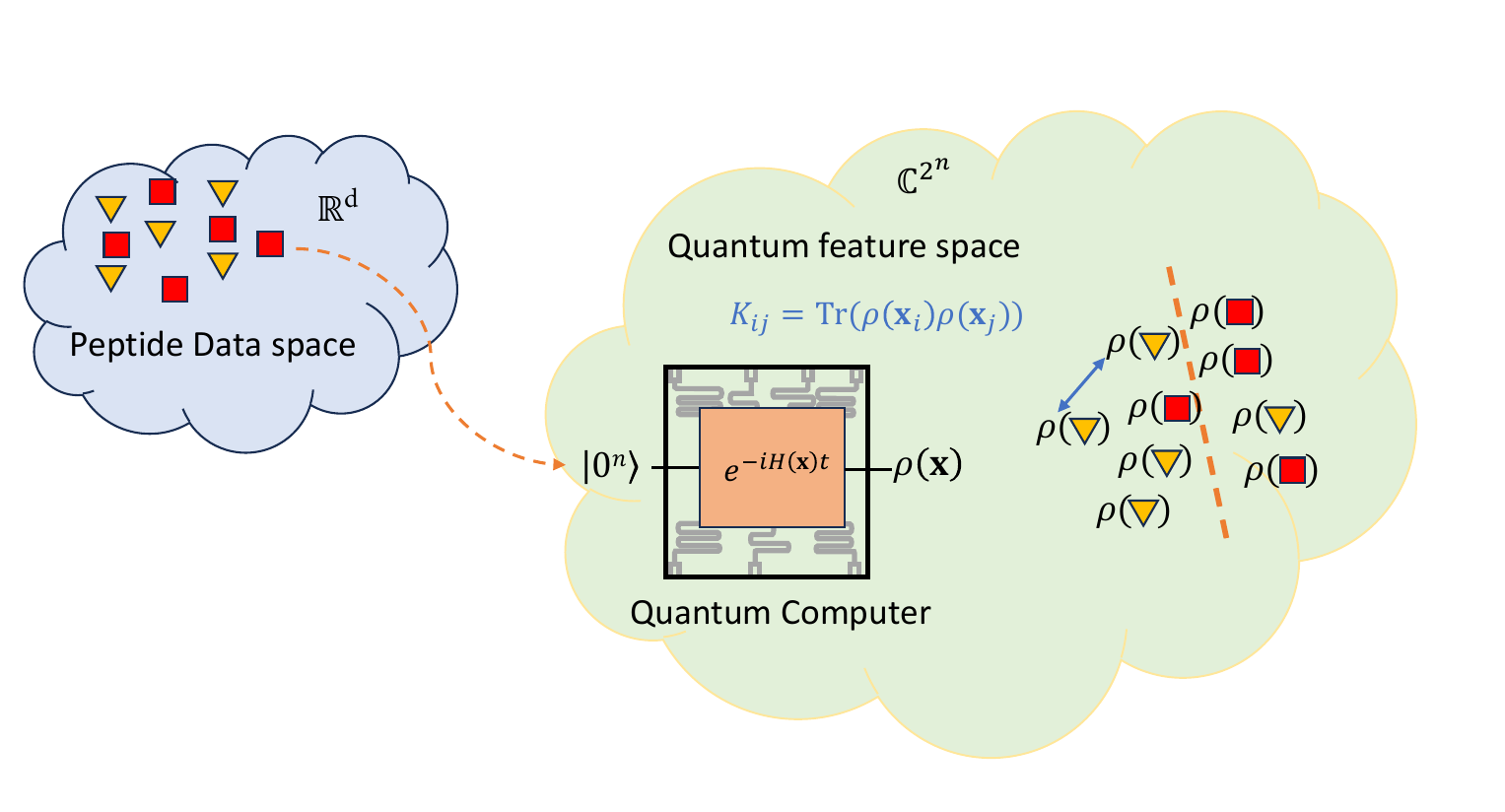}
    \caption{Visual depiction of QKMs}
  \label{fig:workflow_outline}
\end{figure}
    
QKMs~\cite{schuld2021supervised} are powerful tools for handling complex, non-linear relationships in data by exploiting features of quantum mechanics to identify patterns which classical ML algorithms struggle to learn. At the heart of QKMs lies the concept of a quantum feature map, which takes input data and embeds it into a quantum state called a feature state. The feature states belong to a quantum feature space whose dimension scales exponentially in the number of qubits (quantum analogues of bits) and this leads many to believe that QKMs may offer computational advantages over classical kernel methods. In particular, by encoding problem-specific information into the high-dimensional feature states, it might be possible to simplify intricate structures that are challenging to discern in the original input space, even for classical ML models. For example, some problem instances that exhibit highly non-linear relationships in the input space may become linearly separable in the quantum feature space if embedded appropriately~\cite{liu2021rigorous}. 

Let \(\mathcal{D}=\{(\mathbf{x}_i,y_i)\}_{i=1}^{M}\subseteq\mathcal{X}\times\{\pm1\}\) be a training dataset, where
\(\mathcal{X}\equiv\mathbb{R}^d\) is the input data domain,
\(d\in\mathbb{N}\) is the dimension of the input data,
\(\mathbf{x}_i\in\mathcal{X}\) denotes the $i^{\text{th}}$ input training data sample, 
\(y_i\in\{\pm1\}\) denotes the class label for the \(i^{\textrm{th}}\) training data sample and \(M\in\mathbb{N}\) is the total number of training data samples.
A \emph{data encoding unitary}, also sometimes called a parameterised quantum circuit ~\cite{benedetti2019parameterised}, is a function which assigns to each element \(\mathbf{x}\in\mathcal{X}\) a \(2^n\times2^n\) unitary operator \(U(\mathbf{x})\). Given a data encoding unitary \(U\), we define the associated \emph{quantum feature map}~\cite{schuld2021supervised}, denoted \(\rho:\mathcal{X}\to\mathcal{H}_n\), such that 
\begin{equation}
\label{QuantumFeatureMap}
\rho(\mathbf{x})=U(\mathbf{x})|0^n\rangle\langle0^n|U^\dagger(\mathbf{x}),
\end{equation}
where \(\mathcal{H}_n\) is the quantum feature space of \(2^n\times2^n\) Hermitian operators, \(|0^n\rangle\equiv\otimes_{i=1}^{n}|0\rangle\) (where \(\otimes\) denotes the Kronecker product) and \(|0\rangle\equiv\left[\begin{smallmatrix}1\\0\end{smallmatrix}\right]\). The feature state \(\rho(\mathbf{x})\) is the quantum state which results from the physical procedure of applying the unitary operator \(U(\mathbf{x})\) to the initial state \(|0^n\rangle\) on a quantum computer. Given a quantum feature map \(\rho\), the associated \emph{quantum kernel}~\cite{schuld2021supervised}, denoted \(\mathcal{K}_{\rho}:\mathcal{X}\times\mathcal{X}\to\mathbb{R}\), is then defined by
\begin{equation}
\label{QuantumKernel}
\mathcal{K}_{\rho}(\mathbf{x},\mathbf{x}^\prime)=\textrm{Tr}[\rho(\mathbf{x})\rho(\mathbf{x}^\prime)].
\end{equation}
The value of the quantum kernel \(\mathcal{K}_{\rho}(\mathbf{x},\mathbf{x}^\prime)\) is equal to the inner-product in \(\mathcal{H}_n\) between the feature states \(\rho(\mathbf{x})\) and \(\rho(\mathbf{x}^\prime)\) and serves as a measure of similarity between the two states as shown in \cref{fig:workflow_outline}. By substituting \eqref{QuantumFeatureMap} into \eqref{QuantumKernel}, the value of the quantum kernel can equivalently be written as
\begin{equation}
\label{QuantumKernelSimplified}
\mathcal{K}_\rho(\mathbf{x},\mathbf{x}^\prime)=\big|\,\langle0^n|U^\dagger(\mathbf{x})U(\mathbf{x}^\prime)|0^n\rangle\,\big|^2.
\end{equation}
Determining the numerical value of the quantum kernel for different inputs on a classical computer is, in general, a prohibitively expensive operation even for the most powerful supercomputers. This is because evaluating \eqref{QuantumKernelSimplified} (which is less computationally intensive than evaluating \eqref{QuantumKernel}) involves computing the complex-valued vectors \(U(\mathbf{x})|0^n\rangle\), which are of length \(2^n\). So as \(n\) grows, the amount of memory required to simply store the vectors on a classical computer quickly becomes infeasible. However, assuming that the data encoding unitary can be implemented in \(\mathcal{O}(\text{poly}(n))\) time on a quantum computer, we can evaluate the quantum kernel efficiently by applying \(U^\dagger(\mathbf{x})U(\mathbf{x}^\prime)\) to the initial state \(|0^n\rangle\) and running the experiment several times to determine the probability of observing the measurement outcome \(|0^n\rangle\), which is equal to the value of \eqref{QuantumKernelSimplified} \cite{nielsen2000quantum}. 

Not \emph{every} quantum kernel will be classically hard to evaluate, for example, if the data encoding unitary only involves Clifford gates~\cite{aaronson2004simulation} or only entangles a few scattered qubits, 
but there certainly exist quantum kernels which are expected to be classically hard to estimate~\cite{havlicek2019supervised, schuld2019quantum,liu2021rigorous, wu2023quantum}.  This limits classical ML algorithms, including support vector machines, from implementing classification algorithms which might be facilitated by access to such kernel functions. However, with the help of a quantum computer, computing this diverse range of kernels becomes feasible and may facilitate successful ML.
%In this paper, we aim to develop a quantum support vector machine for non-hemolytic peptide classification. 

\subsection{Quantum Support Vector Machines for Non-linear Classification}

Support vector machines (SVMs) are one of the most effective algorithms in machine learning, particularly for non-linear classification tasks. The motivation behind SVMs lies in their emphasis on maximising the margin between different classes in high-dimensional feature spaces. This promotes better generalisation to unseen data and increased resilience to noise in the training dataset. Such an approach enables SVMs to handle both linear and non-linear relationships through the use of kernels, providing a versatile and practical technique for a wide range of applications. Furthermore, by pairing this classical ML approach with the efficient evaluation of quantum kernels on quantum computers, we arrive at the QSVM~\cite{havlicek2019supervised}. 

A QSVM, when supplied with a training dataset \(\mathcal{D}=\{(\mathbf{x}_i,y_i)\}_{i=1}^{M}\subseteq\mathcal{X}\times\{\pm1\}\) and a quantum feature map \(\rho:\mathcal{X}\to\mathcal{H}_n\), attempts to solve the soft-margin dual optimisation problem:
\begin{equation}
\begin{aligned}
\label{DualSVM}
&\min_{{\bm\alpha}\in[0,C]^M}\frac{1}{2}\sum_{i,j}^{M}{\alpha}_i{\alpha}_jy_iy_jK_{ij}-\sum_{i=1}^{M}{\alpha}_i\\
&\qquad\qquad\textrm{s.t.}\quad\sum_{i=1}^{M}{\alpha}_iy_i=0,
\end{aligned}
\end{equation}
where \(C\geq0\) is a parameter which quantifies the penalty associated with a misclassified data point, \(K_{ij}=\mathcal{K}_{\rho}(\mathbf{x}_i,\mathbf{x}_j)\) is an \(M\times M\) matrix called the \emph{quantum kernel matrix} and the solution is the vector \({\bm\alpha}=(\alpha_1,\ldots,\alpha_M)\in[0,C]^M\). Intuitively, solving \eqref{DualSVM} corresponds to maximising the margin between the classes in the quantum feature space, while allowing for the misclassification of some data points at a cost proportional to \(C\). This means that a solution is permitted even when the classes are not linearly separable after being embedded in \(\mathcal{H}_n\).  Once the minimisation problem \eqref{DualSVM} has been solved (which can be achieved deterministically by exploiting some techniques from convex optimisation~\cite{boyd2004convex}) we can make predictions about which class an unseen data sample \(\mathbf{x}\in\mathcal{X}\) belongs to using the equation
\begin{equation}
\label{SVMModel}
y(\mathbf{x}) = \textrm{sign}\left(\sum_{i=1}^{M}\alpha_iy_i\mathcal{K}_{\rho}(\mathbf{x},\mathbf{x}_i)+b\right),
\end{equation}
where \(b\in\mathbb{R}\) can be determined with the Karush-Kuhn-Tucker conditions~\cite{kuhn1951nonlinear}. 
\begin{algorithm}
\caption{Training a QSVM}
\label{alg1}
\textbf{Input:} The training dataset \(\mathcal{D}=\{(\mathbf{x}_i,y_i)\}_{i=1}^{M}\), the data encoding unitary \(U\) which defines \(\rho\) according to \eqref{QuantumFeatureMap}\\
\textbf{Output:} The solution \(({\alpha}_1,\ldots,{\alpha}_M)\in[0,C]^M\) to \eqref{DualSVM}  
\begin{algorithmic}[1]
\For{\(i=1,\ldots,M\)}
   \State Set \(K_{ii}:=1\)
\EndFor
\For{\(i=1,\ldots,M\)}
    \For{\(j=i+1,\ldots,M\)}
        \State Apply \(U^\dagger(\mathbf{x}_i)U(\mathbf{x}_j)\) to the initial state \(|0^n\rangle\)
        \State Measure the probability \(p\) of the outcome \(|0^n\rangle\)
        \State Set \(K_{ij}:=p\) and \(K_{ji}:=p\)
    \EndFor
\EndFor
\State Solve the soft-margin dual optimisation problem \eqref{DualSVM}
\State \Return \(({\alpha}_1,\ldots,{\alpha}_M)\)
\end{algorithmic}
\end{algorithm}

From the description above, one can clearly see that applying the QSVM to a binary classification problem requires us to determine the quantum kernel matrix entries \(K_{ij}=\mathcal{K}_\rho(\mathbf{x}_i,\mathbf{x}_j)\) for all \(i,j=1,\ldots,M\), which is usually done on a quantum computer. There are, however, some shortcuts. For example, as can be seen in \eqref{QuantumKernel}, the quantum kernel matrix is symmetric and has 1's along the main diagonal so we really only need to determine \(K_{ij}\) for \(i<j\) and then symmetrically fill the matrix in order to solve \eqref{DualSVM} (see Algorithm \ref{alg1}). Then once we have determined the solution to \eqref{DualSVM}, to make predictions on a new data sample \(\mathbf{x}\in\mathcal{X}\), we need to estimate the real numbers \(\mathcal{K}_\rho(\mathbf{x},\mathbf{x}_i)\) for all \(i=1,\ldots,M\) on a quantum computer and evaluate \eqref{SVMModel} classically (see Algorithm \ref{alg2}).

\begin{algorithm}
\caption{Making predictions with a trained QSVM}
\label{alg2}
\textbf{Input:} The training dataset \(\mathcal{D}=\{(\mathbf{x}_i,y_i)\}_{i=1}^{M}\), a new input \(\mathbf{x}\in\mathcal{X}\), the solution \(({\alpha}_1,\ldots,{\alpha}_M)\) to \eqref{DualSVM}\\
\textbf{Output:} The class label \(y\in\{\pm1\}\) for the input \(\mathbf{x}\in\mathcal{X}\)
\begin{algorithmic}[1]
\State Set \(c:=0\)
\For{\(i=1,\ldots,M\)}
   \State Apply \(U^\dagger(\mathbf{x})U(\mathbf{x}_i)\) to the initial state \(|0^n\rangle\)
   \State Measure the probability \(p\) of the outcome \(|0^n\rangle\)
   \State Set \(c:=c+{\alpha}_iy_ip\)
\EndFor
\State Determine \(b\) (from KKT conditions)
\State Set \(c:=c+b\)
\State Set \(y:=\textrm{sign}(c)\)
\State \Return \(y\)
\end{algorithmic}
\end{algorithm}

\subsection{Hamiltonian and Quantum Kernel Selection}
\label{sec:encoding}
We now discuss the quantum kernels used in our numerical simulations, together with details about how one would physically implement them on a quantum computer. We begin by defining the set of \(n\)-qubit Pauli strings,
\begin{equation}
\label{PauliStrings}
\mathcal{P}_n\equiv\big\{\otimes_{i=1}^{n}\sigma_i:\sigma_i\in\{\mathbb{I},\sigma_X,\sigma_Y,\sigma_Z\}\big\},
\end{equation}
where \(\sigma_X,\sigma_Y,\sigma_Z\) are the Pauli \(X,Y,Z\) operators respectively~\cite{nielsen2000quantum} and \(\mathbb{I}\) is the \(2\times2\) identity operator. Note that for a given number of qubits, \(n\), there are \(4^n\) different Pauli strings and each can be represented by a different \(2^n\times2^n\) complex-valued matrix. Given an input data sample \(\mathbf{x}=(\mathbf{x}^{(1)},\ldots,\mathbf{x}^{(d)})\in\mathcal{X}\), we randomly sample \(d\) Pauli strings \(\{P_j\}_{j=1}^{d}\subseteq\mathcal{P}_n\) and define the Hermitian operator 
\begin{equation}
\label{ProblemHamiltonian}
H(\mathbf{x})\equiv\sum_{j=1}^{d}\mathbf{x}^{(j)}P_j.
\end{equation} 
In quantum mechanics the Hermitian operator in \eqref{ProblemHamiltonian} is called a \emph{Hamiltonian} and is closely related to the total energy of a physical system. Each term \(\mathbf{x}^{(j)}P_j\) in the Hamiltonian describes some kind of physical interactions that the system undergoes, for example, via the exchange of heat or kinetic energy with another object. By precisely controlling a physical system, e.g. a quantum computer, and the interactions it undergoes we can encode information such as \(\mathbf{x}\in\mathcal{X}\) into the way the system evolves in time, allowing us to implement specific operations such as the matrix exponential of the Hamiltonian in \eqref{ProblemHamiltonian}. The data encoding unitary we use in this work is defined by
\begin{equation}
\label{ProblemDataEncodingUnitary}
U(\mathbf{x})\equiv\left(\prod_{j=1}^{d} e^{-i\mathbf{x}^{(j)}P_jt/s}\right)^s,
\end{equation}
which in turn defines the quantum feature map and the quantum kernel according to \eqref{QuantumFeatureMap} and \eqref{QuantumKernel} respectively. Note that each choice of the Pauli strings \(\{P_j\}_{j=1}^{d}\), number of qubits \(n\in\mathbb{N}\), \(t>0\), and \(s\in\mathbb{N}\) defines a different data encoding unitary and hence a different quantum kernel. 

At a glance, the data encoding unitary in \eqref{ProblemDataEncodingUnitary} looks convoluted and artificial but it is actually motivated by considerations of how one would physically apply a unitary operator on a quantum computer. In particular, the unitary operator we would like to implement is \(e^{-iH(\mathbf{x})t}\), which corresponds to evolving the system according to the interactions described by \(H(\mathbf{x})\) for a time \(t\). Notice that if the Pauli strings \(P_j\) were real numbers (instead of matrices) then \eqref{ProblemDataEncodingUnitary} would simplify to \(e^{-iH(\mathbf{x})t}\). However, when exponentiating matrices it is not necessarily true that \(e^Ae^B\) equals \(e^{A+B}\) so we do not get the same kind of simplification. This is unfortunate because implementing matrix exponentials of individual Pauli strings (and their scalar multiples) is efficient on a quantum computer, so if \(e^Ae^B=e^{A+B}\) was true for all square matrices \(A,B\), then applying \(e^{-iH(\mathbf{x})t}\) would simply involve implementing \(e^{-i\mathbf{x}^{(j)}P_jt}\) for all \(j=1,\ldots,d\) in any order. But alas, things are not so simple. Instead we use an approximation inspired by the Trotter product formula~\cite{nielsen2000quantum} which states that
\begin{equation}
\label{TrotterFormula}
e^{A+B}=\lim_{s\to\infty}\left[e^{A\slash s}e^{B\slash s}\right]^s.
\end{equation}
Note that we can introduce \(t\) into \eqref{TrotterFormula} by replacing \(A\mapsto At\) and \(B\mapsto Bt\). Ideally, we would like to take \(s\to\infty\) so that the equality holds, but this would translate into a physical procedure that takes an infinite amount of time. So, instead we choose a reasonable value for \(s\in\mathbb{N}\) to get a sensible approximation. Iterating the approximation for each term in the exponent of \(e^{-iH(\mathbf{x})t}\) results in \eqref{ProblemDataEncodingUnitary}. So, \eqref{ProblemDataEncodingUnitary} is really just a description of the experimental implementation of \(e^{-iH(\mathbf{x})t}\). Also note that the encoding of data into the Hamiltonian in \eqref{ProblemHamiltonian} could be used for other quantum algorithms too.

Another nice aspect of using a data encoding unitary which is closely related to \(e^{-iH(\mathbf{x})t}\) is that we can derive an upper bound on the generalisation error that the associated QSVM will exhibit, which helps us to choose appropriate values of \(t\) (and \(s\)).

\begin{theorem}
\label{theorem_gen_error_bound}
    Consider the training dataset $\mathcal{D}=\{(\mathbf{x}_i,y_i)\}_{i=1}^M\subseteq\mathcal{X}\times\{\pm1\}$, where \(\mathcal{X}\equiv\mathbb{R}^d\). Let $U(\mathbf{x})=e^{-iH(\mathbf{x})t}$ be the data encoding unitary, where the Hermitian operator $H(\mathbf{x})\equiv\sum_{j=1}^{d}\mathbf{x}^{(j)}P_j$ is defined in terms of $n$-qubit Pauli strings $\{P_j\}_{j=1}^{d}\subseteq\mathcal{P}_n$, and the evolution time $t$ satisfies $t\ll 1$. Then for any $\delta>0$, with probability $\geq1-\delta$, the generalised error, $\epsilon$, of quantum kernel method satisfies
    \begin{align}
        \epsilon\leq\frac{8(\|\bm\alpha\|^2+\kappa(\mathcal{X})t^2)}{\sqrt{M}}\left(1+\frac{1}{2}\sqrt{\frac{\log(1/\delta)}{2}}\right),
    \end{align}
    where $\kappa(\mathcal{X}) = \sum_{i,j}\alpha_i\alpha_j[\langle0^n|(H(\mathbf{x}_i) - H(\mathbf{x}_j))|0^n\rangle]^2$.
\end{theorem}

A proof of this theorem is given in \ref{app:proof}. This result implies that the generalisation error will be proportional to \(t^2\), meaning that a smaller value for \(t\) will likely be better. However one should keep in mind that an upper-bound does not give us an equality, just some basic insight which may guide the choice of our hyperparameters. The real deciding factor which will influence the optimal hyperparameter values comes from the actual model performances. 

\section{Experiments and Discussion}
In this section, we provide details about the datasets, how they were preprocessed, and the associated learning task considered in this work. We then discuss how we chose the hyperparameters that define the classical and quantum kernel models which are applied to the testing set.

\subsection{Datasets and Learning Task}
\label{hemopi}
In this work we use the three HemoPI datasets: HemoPI-1, HemoPI-2, and HemoPI-3. These datasets are compilations of experimentally validated peptides labelled as hemolytic or non-hemolytic which have been extracted from various sources, including the Hemolytik database~\cite{gautam2014hemolytik}, Swiss-Prot~\cite{jungo2012uniprotkb} (a curated protein sequence database which is part of UniProt~\cite{UniProt2023}), and the Database of Antimicrobial Activity and Structure of Peptides (DBAASP)~\cite{pirtskhalava2016dbaasp}. The datasets were originally published by Chaudhary et al.~\yrcite{chaudhary2016web} and are freely available for download on the HemoPI website.
HemoPI-1 is perfectly balanced and comprises 552 hemolytic peptides from Hemolytik and 552 non-hemolytic peptides from Swiss-Prot. HemoPI-2 contains 552 peptides with high hemolytic efficiency and 462 non-hemolytic peptides from Hemolytik. HemoPI-3 consists of 885 peptides with high hemolytic efficiency and 738 low/non-hemolytic peptides from Hemolytik and DBAASP.
The peptide sequence lengths range from 4 to 133 amino-acids residues, with a mean of 20 across the 3 datasets. All three datasets will made available on GitHub upon publication. 

The standard representation of a protein sequence is given by a string of alphabetic letters, where each letter represents an amino acid. There are many choices one might make about how to convert the alphabetic sequences into vectors and this choice is of vital importance since it ultimately determines what can be learned by the ML model. One may also choose to represent a peptide in other fashions, for example, commonly used embeddings include one-hot-encoding and AAIndex~\cite{kawashima2007aaindex}. However, one issue associated with choosing an embedding is that peptides have varying lengths, and the consequences of having different input lengths for a learning model can be quite severe. To overcome this issue we chose to use 40 physicochemical descriptors, such as the molecular weight and hydrophobicity which are calculated through experimentally measured properties, to represent each peptide sequence.  Refer to \ref{app:descriptors} for detailed description of the physicochemical properties utilised in this study. Accordingly, the learning task associated with the datasets involves classifying peptide sequences as either \textit{hemolytic} or \textit{non-hemolytic} based on input data given by the 40-dimensional vectors of physicochemical descriptors. With respect to the data encoding unitary \eqref{ProblemDataEncodingUnitary}, this means that we have \(d=40\) and that \(\mathcal{X}\equiv\mathbb{R}^{40}\).

\subsection{Preprocessing}
The HemoPI-1 dataset is balanced in terms of class distribution, while HemoPI-2 and HemoPI-3 were slightly imbalanced. Accordingly, undersampling was applied to the training datasets of HemoPI-2 and HemoPI-3 to ensure that each dataset individually contains the same number of instances for both classes. Each dataset is then divided into two subsets: a larger training dataset (80\%) used for model development and a smaller test dataset (20\%) for evaluating the model's performance on unseen data. Additionally, the training and testing datasets from each of the 3 HemoPI datasets were standardised. In particular, we applied \(z\)-score standardisation to the distribution of values for each physicochemical descriptor within each of the 3 datasets according to the training set. This means that, after preprocessing, in each of the 3 training datasets each physicochemical descriptor individually takes on numerical values which are distributed with a mean of 0 and a standard deviation of 1. The test sets were then adjusted in a similar fashion according to the mean and standard deviations calculated in the training sets.

\subsection{Model Selection and Implementations}
In this work we apply the QSVM with the feature map defined by \eqref{ProblemDataEncodingUnitary}, as well as popular classical kernels such as the linear, polynomial, and radial basis function, % and sigmoid kernel functions, 
to each of the three HemoPI datasets. In order to select the best hyperparameter values for all the different models, we perform a stratified 5-fold cross-validation and a grid-search over the different hyperparameter values and pick the best model  for each dataset according to the metric of average accuracy over the 5 folds. 

For the quantum kernel, we need to choose an appropriate number of qubits \(n\in\mathbb{N}\), a value for \(t>0\) and a value for \(s\in\mathbb{N}\). For the number of qubits $n$, we simulated up to $n=14$ and found no obvious improvement in validation accuracy for $n > 6$. As a large $n$ value inevitably lengthens the training time, we set \(n=6\) for the remainder of the simulations. In order to pick \(t\) and \(s\), we carried out 5-fold cross-validation for different choices of $(t,s)$. Figure \ref{fig:train_val_acc} shows the average training and validation accuracies from the cross validation on each dataset. According to these results we picked the optimal hyperparameter values for each dataset (indicated by red arrows in Figure \ref{fig:train_val_acc}; see also \cref{tab:best}). Note that the generalisation error increases as $t$ increases, while $s$ does not contribute significantly to the accuracy, which is consistent with \cref{theorem_gen_error_bound}.  

Similarly, for the classical kernels for SVM, we trained the SVM classifiers with the linear, RBF, and polynomial kernels using 5-fold cross validation. For all the HemoPI datasets, the classifiers with the RBF kernel consistently outperformed those with the linear and polynomial kernels. The optimal hyperparameters for SVM with the RBF kernel are shown in \cref{tab:best}; the optimal hyperparameters for the other two kernels are shown in \ref{app:linear_poly}.

\begin{table}[t!]
\caption{The best QSVM and classical SVM for the HemoPI datasets, where $n$ denotes the number of qubits; $t$ and $s$ are the variables described in \eqref{ProblemDataEncodingUnitary}; $C$ is the regularisation parameter for the SVM classifier; and $\gamma$ is the SVM kernel parameter.}
\label{tab:best}
\vskip 0.15in
\begin{center}
\small
\begin{tabular}{l|ccr|ccc}
\toprule
& \multicolumn{3}{c|}{QSVM} & \multicolumn{3}{c}{Classical SVM} \\
\cmidrule{2-7}
Dataset &$n$& $t$ & $s$ & Kernel & $C$ & $\gamma$ \\
\midrule
HemoPI-1 &6 & 0.3 & 10 & RBF & 100 & 0.001\\
HemoPI-2 &6 & 0.15   & 10 & RBF & 10 & 0.001\\
HemoPI-3 &6 & 0.15 & 10 & RBF & 10 & 0.01\\
\bottomrule
\end{tabular}
\end{center}
\vskip -0.1in
\end{table}

\begin{figure*}[t!]
    \centering
    \includegraphics[width=0.8\textwidth]{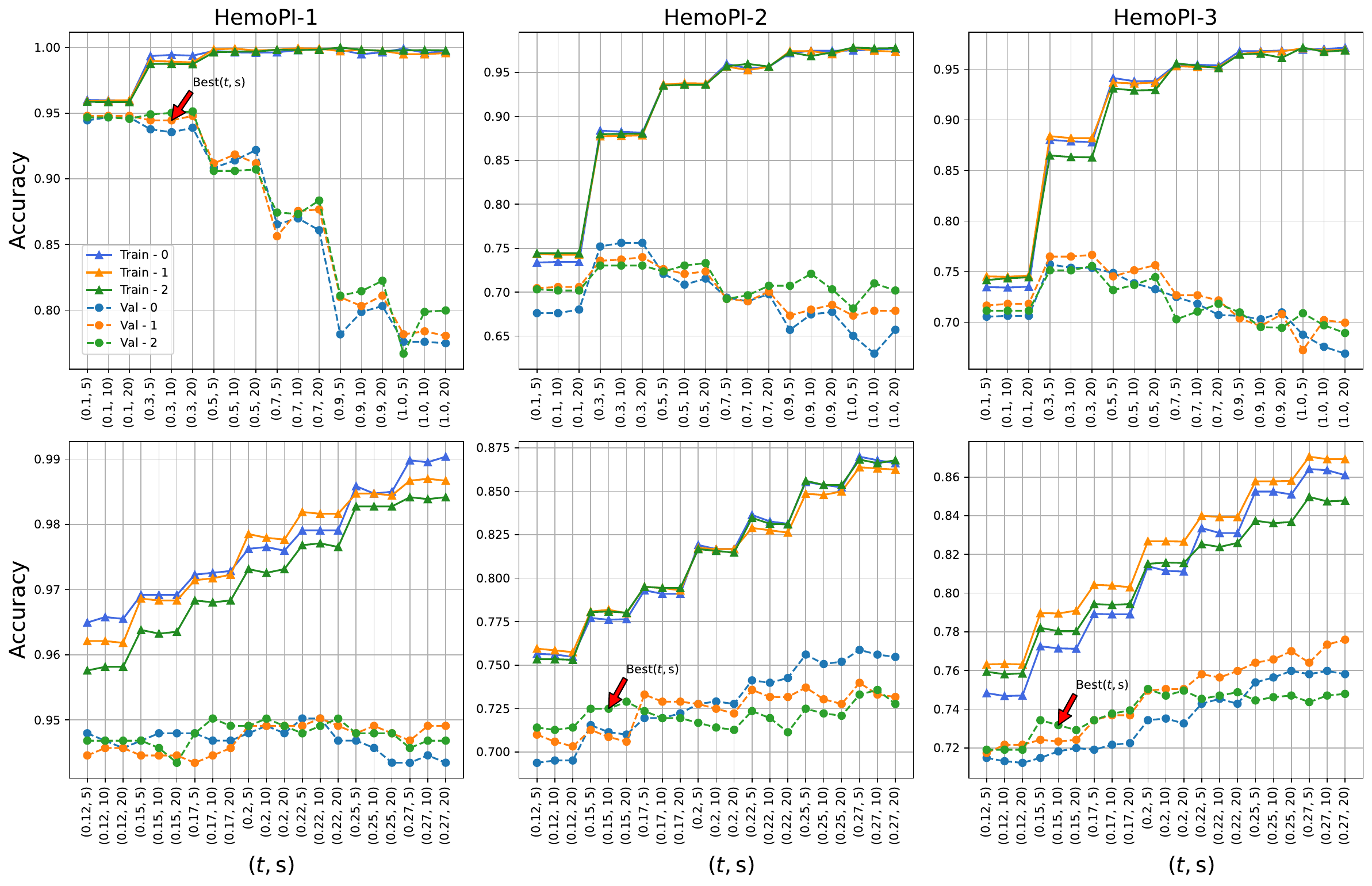}
    \caption{Training and validation accuracies of the quantum kernel on all three HemoPI datasets for different combinations of $(t,s)$ values. This plot illustrates the comparison of training (Train-0, Train-1, Train-2) and validation (Val-0, Val-1, Val-2) accuracies on three sets of randomly sampled Pauli strings for the Hamiltonian, indexed as 0, 1, and 2. 
    The top row shows the initial grid search for the range of $(t,s)$ values; the bottom row shows the final grid search on the $(t,s)$ values found in the top row.}
    \label{fig:train_val_acc}
\end{figure*}

The quantum kernels were implemented with a noiseless classical simulation using the \texttt{\small pennylane} Python library~\cite{bergholm2018pennylane}. In particular, \texttt{\small pennylane} was used to calculate the quantum kernel matrix which was then passed to the \texttt{\small scikit-learn} function \texttt{\small svm.svc} as a \emph{precomputed} kernel. The classical kernels were also implemented with \texttt{\small scikit-learn}'s \texttt{\small svm.svc}. The best models having the highest average accuracies and smallest difference from the average training accuracies were then picked from the cross validation results.

\cref{table1:testing acc} shows the best results we obtained from the quantum kernel and classical kernels. The parameters of all kernels are shown in \cref{tab:best}. As mentioned in \cref{hemopi}, the non-hemolytic peptides in HemoPI-1 were randomly generated from Swiss-Prot, potentially having minimal or no overlap with the hemolytic peptides. This distinct separation between the two classes resulted in high prediction accuracy. However, in HemoPI-2 and HemoPI-3, the non-hemolytic peptides were real peptides that were difficult to distinguish from the hemolytic peptides, resulting in generally lower prediction accuracy. The t-SNE projections show the class distributions of 3 datasets based on the 40 descriptors in \ref{app:tsne} and are consistent with these intuitions.

\begin{table}[t!]
\caption{Testing accuracies on the HemoPI test sets.} 
\label{table1:testing acc}
\vskip 0.15in
\begin{center}
\begin{small}
\begin{sc}
\begin{tabular}{lcccr}
\toprule
Kernel & HemoPI-1 & HemoPI-2 & HemoPI-3 \\
\midrule
Quantum   & 95.5         & \textbf{74.3} & \textbf{76.0} \\
Linear    &94.5 & 72.3             & 75.1 \\
RBF       &\textbf{96.8} & 72.3   & 75.1 \\
Poly      &94.5 & 71.2             & 73.2 \\

\bottomrule
\end{tabular}
\end{sc}
\end{small}
\end{center}
\vskip -0.1in
\end{table}

\begin{figure*}[t!]
    \centering
    \includegraphics[width=0.8\textwidth]{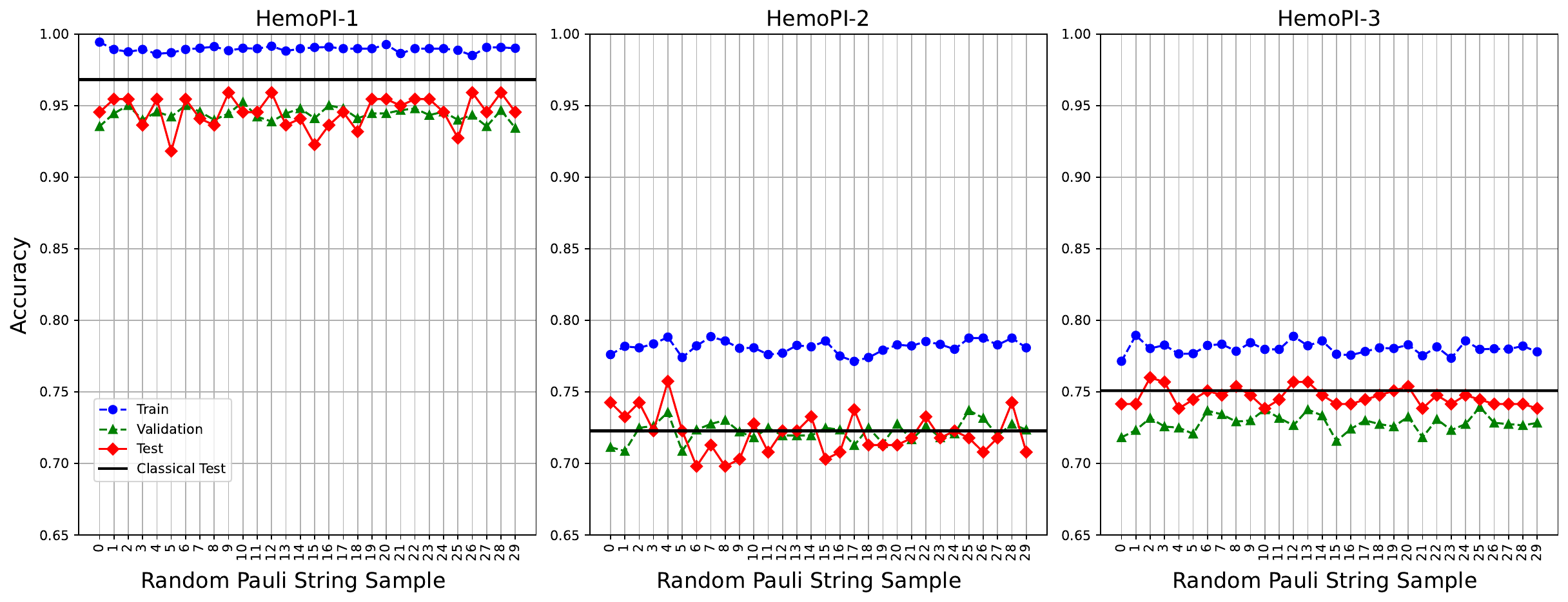}
    \caption{Training, validation and testing accuracies on 30 sets of randomly sampled Pauli Strings with the best $t$ and $s$ values identified in \cref{fig:train_val_acc} for each dataset. The black horizontal line in each plot indicates the testing accuracy from the best classical kernel shown in \cref{table1:testing acc}.}
    \label{fig:30sets}
\end{figure*}

The results in \cref{table1:testing acc} show that quantum kernels achieved higher accuracies on HemoPI-2 and HemoPI-3 compared to classical SVMs, but were outperformed by the classical RBF kernel on HemoPI-1. In particular, an increase in accuracy of 2\% and 0.9\% was achieved on HemoPI-2 and HemoPI-3 respectively compared to classical SVMs. While this increase might be small from a machine learning perspective, such an improvement can have significant clinical and pharmaceutical implications. From a medical perspective this improvement means safer development of therapeutic peptides and a reduced risk of hemolysis and subsequent complications. Similarly, when compared with the results of Plisson et al.~\yrcite{plisson2020machine}, which held the previous record for the best accuracy on the HemoPI datasets, the quantum kernels provide accuracies which improve on their classical ML results by 3.1\%,  2\%, and 2.8\% for the HemoPI-1, HemoPI-2, and HemoPI-3 datasets respectively. Worth noting however is that their models used slightly different representations for the peptides, with 56 physicochemical descriptors instead of the 40 used in this work. Consequently, we cannot rule out that their classification algorithms may perform better than the QSVM if supplied with identical input data.

\subsection{Random Pauli Strings}
In the machine learning community, the idea that feature maps which are tailored for specific problems tend to result in better performance is widely agreed upon and empirically supported. However, our results seem to indicate that for the HemoPI datasets, especially HemoPI-2 and HemoPI-3, that we have a reasonable amount of freedom about which feature maps to use (within a large class). In particular, as can be seen in \cref{fig:train_val_acc}, the feature maps defined by different random samples of the Pauli operators \(\{P_j\}_{j=1}^{40}\) all provide similar performances (in terms of labelling accuracy) which are competitive with the performance of the best classical kernels for this problem. This might indicate that the 40 physicochemical descriptors are a sufficient and effective representation of the peptides when paired with the data encoding unitary in \eqref{ProblemDataEncodingUnitary}. To confirm our speculation, we conducted further experiments with 30 sets of random Pauli string samples with the best $t$ and $s$ identified from \cref{fig:train_val_acc} for each dataset.

In \cref{fig:30sets}, 30 sets of Pauli strings were sampled and the associated QSVMs demonstrated similar training, validation and testing accuracy trajectories. The testing accuracies are on par with the best testing accuracies of classical kernels we obtained. Accordingly, our results could be used as a baseline for designing data encoding unitary for classical data when the optimal data encoding unitary is not apparent.

\section{Conclusion and Future Work }
In this work, we explored applications of quantum computing to hemolytic peptide classification tasks, which are greatly important in clinical and pharmaceutical contexts. We performed the first application of the QSVM to a peptide classification task with three datasets, HemoPI-1, HemoPI-2, and HemoPI-3, and observed that the QSVM achieved higher accuracies than both classical SVMs and the best published results for the same classification task. However, due to the fact that our datasets used different physicochemical descriptors to represent the peptides, we cannot exclude the possibility that other classical ML models could outperform the QSVM if supplied with the 40 physicochemical descriptors used in this work. We believe the extensions of this research which make use of other QML methods for different problem instances are possible and that this work could open up paths that lead to new quantum bioinformatics applications.

As mentioned in \cref{sec:encoding}, the data encoding unitary used in this work could also be used for other quantum algorithms which might be worth investigating in future work. Similarly, the three datasets we used in this work are different but belong to the same domain, so another avenue that future researchers could benefit from is to apply the QSVM to problem instances from different domains, including protein engineering, natural language processing and financial data. Finally, the number of qubits in this work was set to 6, which allows classical simulations but also provides the opportunity to assess the robustness and limitations of experimental implementations of the QSVM on a practical problem.

\bibliography{citation}
\bibliographystyle{abbrvnat}

%%%%%%%%%%%%%%%%%%%%%%%%%%%%%%%%%%%%%%%%%%%%%%%%%%%%%%%%%%%%%%%%%%%%%%%%%%%%%%%
%%%%%%%%%%%%%%%%%%%%%%%%%%%%%%%%%%%%%%%%%%%%%%%%%%%%%%%%%%%%%%%%%%%%%%%%%%%%%%%
% APPENDIX
%%%%%%%%%%%%%%%%%%%%%%%%%%%%%%%%%%%%%%%%%%%%%%%%%%%%%%%%%%%%%%%%%%%%%%%%%%%%%%%
%%%%%%%%%%%%%%%%%%%%%%%%%%%%%%%%%%%%%%%%%%%%%%%%%%%%%%%%%%%%%%%%%%%%%%%%%%%%%%%
\newpage
\appendix
\gdef\thesection{Appendix \Alph{section}}
\onecolumn

\section{Proof of Theorem~\ref{theorem_gen_error_bound}}
\label{app:proof}
We first introduce the following lemma to support our proof.
\begin{lemma}[Theorem~5.5 in Ref.~\cite{mohri2018foundations}]
\label{lemma1}
    Let $\mathcal{K}:\mathcal{X}\times\mathcal{X}\mapsto\mathbb{R}$ be a PDS kernel and let $\Phi:\mathcal{X}\mapsto\mathcal{H}$ be the feature map associated with $\mathcal{K}$. Let a uniformly random subset $S\in\{\mathbf{x}:\mathcal{K}(\mathbf{x},\mathbf{x})\leq r^2\}$ be a sample set of size $M$, and let $\mathcal{H}=\{\mathbf{x}\mapsto \mathbf{w}\cdot \Phi(\mathbf{x}):\|\mathbf{w}\|_{\mathcal{H}}\leq\eta\}$ for some $\eta\geq 0$. Then the Rademacher complexity of \(\mathcal{H}\), denoted \(\hat{R}_S(\mathcal{H})\), satisfies
    \begin{align}
        \hat{R}_S(\mathcal{H})\leq\sqrt{\frac{r^2\eta^2}{M}}.
    \end{align}
\end{lemma}
\begin{proof}[Proof sketch of Theorem~\ref{theorem_gen_error_bound}]
    In the context of the QML models associated with QKMs, we can equivalently define the feature states as
    \begin{align}
        \Phi(\mathbf{x})=|\psi(\mathbf{x})\rangle\otimes |\psi^*(\mathbf{x})\rangle,
    \end{align}
    where $|\psi(\mathbf{x})\rangle=U(\mathbf{x})|0^n\rangle$,
    and the linear combination which defines the hyperplane in the quantum feature space as
    \begin{align}
        \mathbf{w}=\sum\limits_{j=1}^M\alpha_j|\psi(\mathbf{x}_j)\rangle\otimes|\psi^*(\mathbf{x}_j)\rangle.
        \label{Eq:linear_parameter}
    \end{align}
    The correctness of the forms above can easily be checked by $\langle\mathbf{w},\Phi(\mathbf{x})\rangle=\sum_{j=1}^M\alpha_j\|\langle\psi(\mathbf{x})|\psi(\mathbf{x}_j)\rangle\|^2$ which gives the same expression when \(\Phi(\mathbf{x})\) is replaced with \(\rho(\mathbf{x})\) defined in \eqref{QuantumFeatureMap}.

    Without loss of generality, we assume the target function $f$ satisfies $|f|\leq \eta r$. For all $\mathbf{x}\in\mathcal{D}$, we have $|\mathbf{w}\cdot\Phi(\mathbf{x})|\leq \|\mathbf{w}\|_{\mathcal{H}}\|\Phi(\mathbf{x})\|_{\mathcal{H}}\leq\|\mathbf{w}\|_{\mathcal{H}}$, thus, for all $\mathbf{x}\in\mathcal{D}$ and hypothesis $h\in\mathcal{H}$, we have $|h(\mathbf{x})-f(\mathbf{x})|\leq2\|\mathbf{w}\|_{\mathcal{H}}$. 
    With $t \ll 1$, using \eqref{Eq:linear_parameter} and $U(\mathbf{x})=e^{-iH(\mathbf{x})t}$  yields
    \begin{align}
           \|\mathbf{w}\|_{\mathcal{H}}=\sqrt{\langle\mathbf{w},\mathbf{w}\rangle}&=\sqrt{\sum\limits_{i,j}\alpha_i\alpha_j\|\langle0^n|e^{-i(H(\mathbf{x}_i)-H(\mathbf{x}_j))t}|0^n\rangle\|^2} \nonumber\\
           &\approx \sqrt{\sum\limits_{i,j}\alpha_i\alpha_j\left[1+\langle0^n| \left(H(\mathbf{x}_i)-H(\mathbf{x}_j) \right) |0^n\rangle^2t^2\right]} \nonumber\\
           &=\eta,
    \label{Eq:eta}
    \end{align}
    where the second line comes from $e^{-iHt}\approx I-iHt$ when $t\ll 1$.
    
    From the bound on the empirical Rademacher complexity (Lemma~\ref{lemma1}), McDiarmid's inequality~\cite{combes2015extension}, and Theorem~10.3 in Ref.~\cite{mohri2018foundations}, we get
    \begin{align}
    \label{upperbound}
        R(h)\leq\hat{R}(h)+8\eta r\hat{R}_S(\mathcal{H})+(2\eta r)^2\sqrt{\frac{\log(1/\delta)}{2M}}.
    \end{align}
    Finally, substituting $r=1$ (which corresponds to $\mathcal{K}(\mathbf{x},\mathbf{x})\leq1$) and $\eta$ (given in \eqref{Eq:eta}) into the right hand side of \eqref{upperbound}, we get the following upper bound for the generalised error, $\epsilon$:
    \begin{align}
             \epsilon=R(h)-\hat{R}(h)&\leq \frac{8r^2\eta^2}{\sqrt{M}}\left(1+\frac{1}{2}\sqrt{\frac{\log(1/\delta)}{2}}\right) \nonumber\\
              &=\frac{8(\|\bm\alpha\|^2+\kappa(\mathcal{X})t^2)}{\sqrt{M}}\left(1+\frac{1}{2}\sqrt{\frac{\log(1/\delta)}{2}}\right),
    \end{align}
    where $\kappa(\mathcal{X})=\sum_{i,j}\alpha_i\alpha_j \left[\,\langle0^n| \left(H(\mathbf{x}_i)-H(\mathbf{x}_j) \right)|0^n\rangle \,\right]^2$. This concludes the proof.
\end{proof}

% -------------------------------------

\newpage
\section{40 Physicochemical Descriptors}
\label{app:descriptors}
The property names and relevant references for the 40-dimensional physicochemical descriptors used in our experiments are shown in the table below:
\begin{center}
    \includegraphics[width=0.85\linewidth]{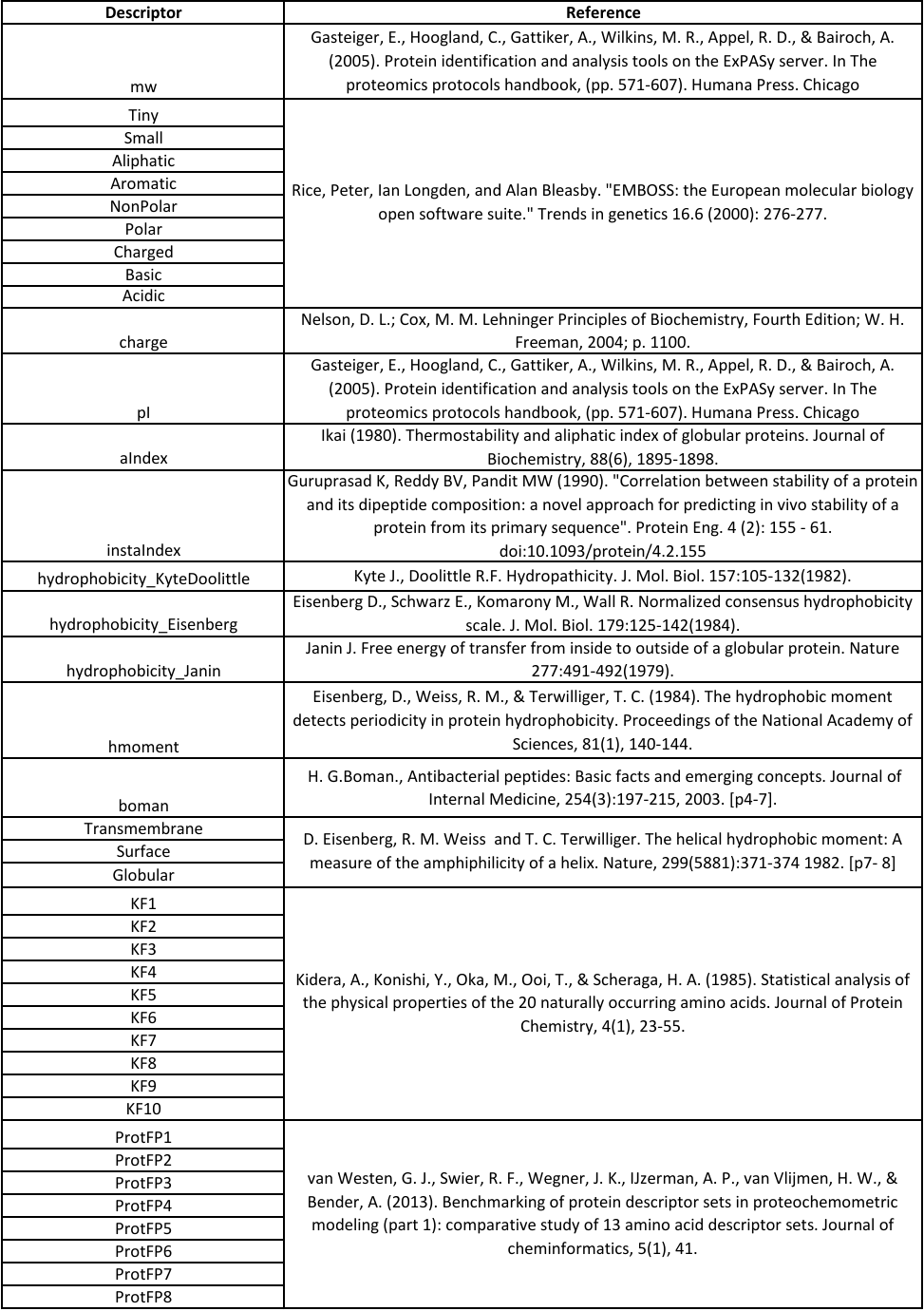}
\end{center}

% -------------------------------------

\newpage
\section{Optimal Hyperparameters for the Linear and Polynomial Kernels}
\label{app:linear_poly}
The table below shows the optimal hyperparameter values for the classical SVMs with linear and polynomial kernels. Here, $C$ denotes the regularisation parameter for the SVM classifier; and deg denotes the degree of the polynomial kernel.
\begin{center}
\small
\begin{tabular}{l|c|cc}
\toprule
& \multicolumn{1}{c|}{Linear} & \multicolumn{2}{c}{Poly} \\
\cmidrule{2-4}
Dataset &$C$  & $C$ & deg \\
\midrule
HemoPI-1 &1   &  0.0001 & 2\\
HemoPI-2 &1000   &  1 & 2\\
HemoPI-3 &10  &  10 & 2\\
\bottomrule
\end{tabular}
\end{center}

% -------------------------------------

\section{t-SNE Plots of the HemoPI Datasets}
\label{app:tsne}
The t-SNE plots below show the 2D visualisation after the 40-dimensional peptide descriptors are projected to 2D using the t-SNE algorithm. It is clear that the two classes are almost linearly separable for HemoPI-1 (with a few obvious outlying points) but there are no clear patterns of clusters for the two classes in HemoPI-2 and HemoPI-3.

\begin{center}
\begin{minipage}[t]{0.3\textwidth}
        \includegraphics[width=\linewidth]{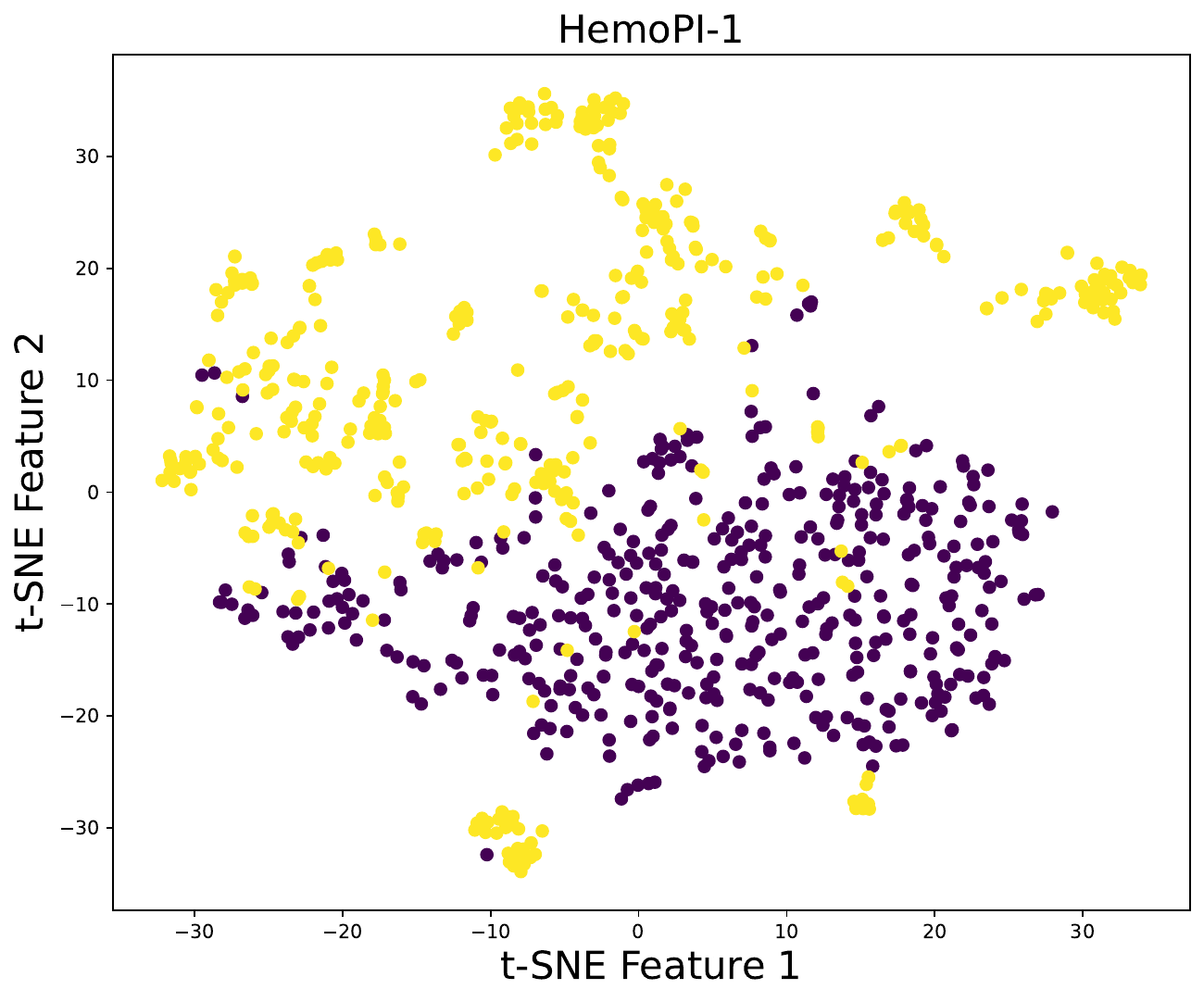}
\end{minipage}
~
\begin{minipage}[t]{0.3\textwidth}
        \includegraphics[width=\linewidth]{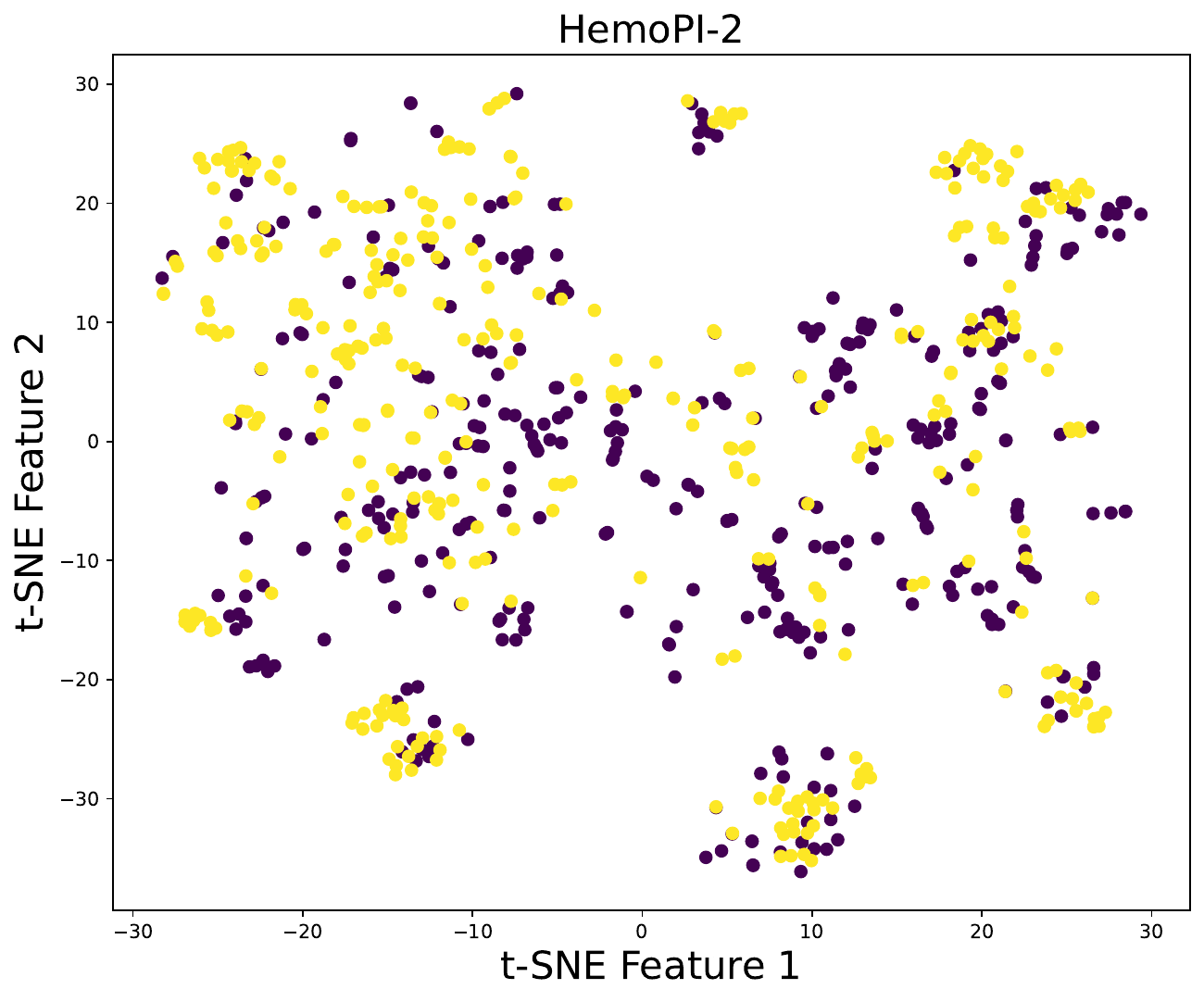}
\end{minipage}
~
\begin{minipage}[t]{0.3\textwidth}
        \includegraphics[width=\linewidth]{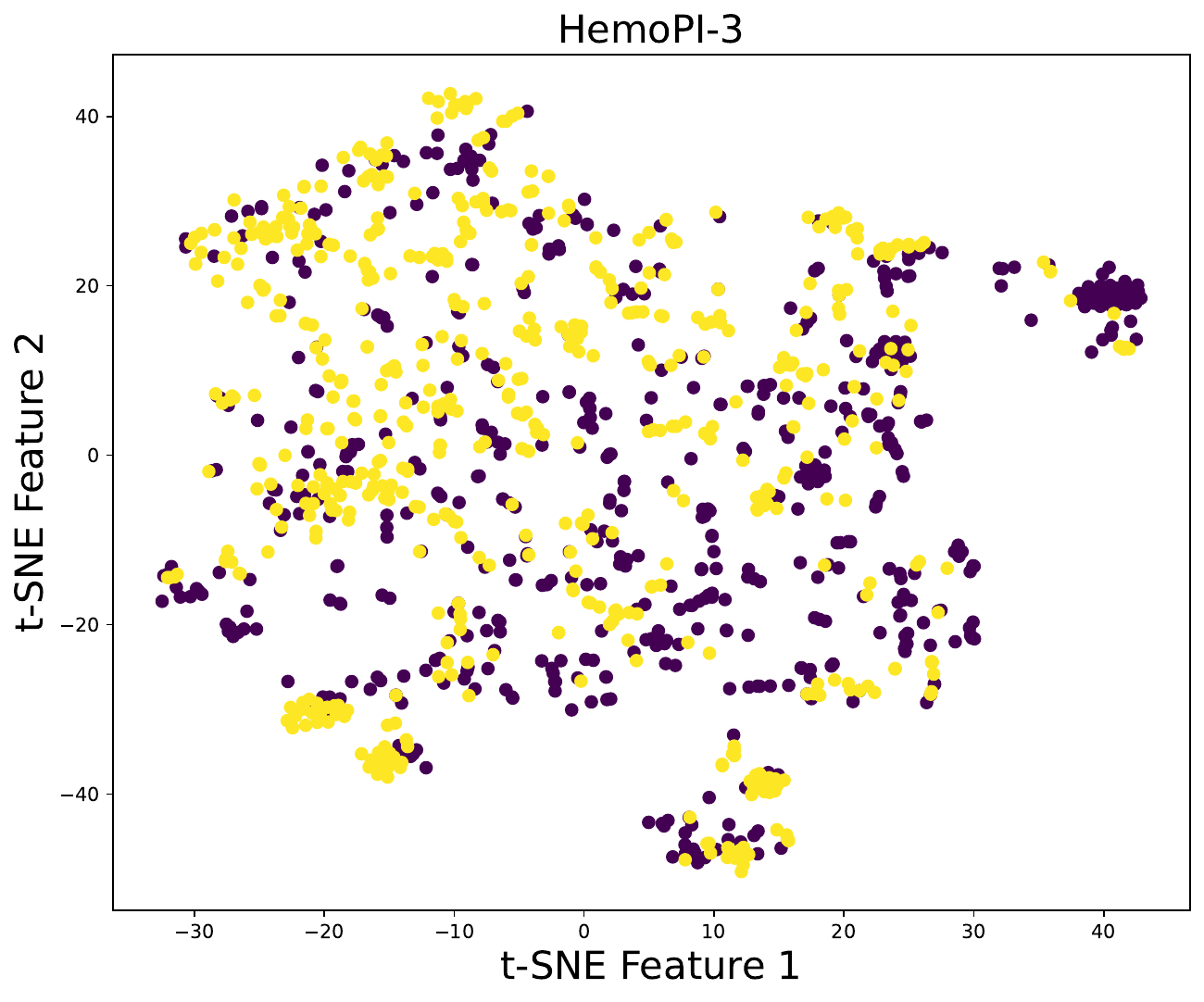}
\end{minipage}
\end{center}

% -------------------------------------

\end{document}